\documentclass[11pt]{article}
\usepackage{amsmath, amsthm, amssymb, eucal, setspace, mathrsfs}
\usepackage[all]{xy}
\usepackage{graphicx}

\newtheorem{theorem}{Theorem}[section]
\newtheorem*{mymeta}{Meta-Statement}

\newtheorem{proposition}[theorem]{Proposition}
\newtheorem{conjecture}[theorem]{Conjecture}

\theoremstyle{definition}

\newtheorem{remark}[theorem]{Remark}
\newtheorem{example}[theorem]{Example}
\numberwithin{equation}{section}

\newcommand{\Hom}{\operatorname{Hom}}

\title{\textbf{Gauge theory and mirror symmetry}}

\author
{Constantin Teleman\thanks{The author thanks the MSRI for its hospitality during 
the writing of this paper. The work was partially supported by NSF grant DMS-1007255.}}

\begin{document}



\maketitle

\begin{abstract}
\noindent
Outlined here is a description of \emph{equivariance} in the world of $2$-dimensional 
extended topological quantum field theories, under a topological action of compact 
Lie groups. In physics language, I am gauging the theories --- coupling them to a principal 
bundle on the surface world-sheet. I describe the data needed to gauge the theory, as 
well as the computation of the gauged theory, the result of integrating over all bundles. 
The relevant theories are `$A$-models', such as arise from the Gromov-Witten theory of a 
symplectic manifold with Hamiltonian group action, and the mathematical description starts 
with a group action on the generating category (the Fukaya category, in this example) which 
is factored through the topology of the group. Their mirror description involves 
holomorphic symplectic manifolds and Lagrangians related to the Langlands dual group. An 
application  recovers the complex mirrors of flag varieties proposed by 
Rietsch.
\end{abstract}

\section{Introduction}

This paper tells the story of equivariance, under a compact Lie group, in the higher algebra 
surrounding \emph{topological quantum field theory} (TQFT). Speaking in riddles, if
$2$-dimensional TQFT is a higher analogue of cohomology (the reader may think of the Fukaya-Floer 
theory of a symplectic manifold as refining ordinary cohomology), my story of gauged 
TQFTs is the analogue of equivariant cohomology. The case of finite groups, well-studied in the 
literature \cite{tur}, provides a useful and easy reference point, but the surprising features 
of the continuous case, such as the appearance of holomorphic symplectic spaces and Langlands 
duality, are missing there. 

From another angle, this is a story of the categorified representation theory of a compact Lie 
group $G$, with the provision that representations are \emph{topological}: the $G$-action 
(on a linear category) factors through the topology of $G$. One floor below, where the group acts 
on vector spaces, these would be not the ordinary complex representations of $G$, but the 
local systems of vector spaces on the classifying space $BG$. There is no 
distinction for a finite group, but in the connected case, $BG$ is simply connected, and we 
must pass to the derived category to see anything interesting. The same will hold in the 
categorified story, where simply connected groups will appear to have trivial representation 
theory, before deriving. This observation suggests a 
straightforward homological algebra approach to the investigation, worthy of featuring as 
an example in a graduate textbook. Pursuing that road, however, leads to faulty predictions, 
even in the simplest case of pure gauge theory of a point (topological Yang-Mills theory). 
One reason for this failure is a curious predilection of interesting TQFTs to break the obvious 
$\mathbb{Z}$-grading information present, collapsing it to a $\mathbb{Z}/2$ grading, 
or encoding it in more labored form (as in the Euler field of Gromov-Witten theory \cite{man}). 
The result is that homological algebra,  which localizes the spectrum of a graded ring to 
its degree zero part, loses relevant information, which needs restoration by ulterior guesswork. 
In our example, we will see the homological information in the neighborhood of a Lagrangian 
within a certain holomorphic symplectic manifold, whereas most of the interesting `physics' 
happens elsewhere.

The emerging geometric picture for this categorical topological representation theory 
is surprisingly attractive. Representations admit a character theory, 
but characters are now coherent sheaves on a manifold related to the conjugacy classes, instead 
of functions. The  manifold in question, the \emph{BFM space} of the Langlands dual Lie group 
$G^\vee$, introduced in \cite{bfm}, is closely related to the cotangent bundle to the space of conjugacy 
classes in the complex group $G^\vee_{\mathbb{C}}$. (For $\mathrm{SU}_2$, it is the Atiyah-Hitchin 
manifold studied in detail in \cite{athit}.) Multiplicity spaces of $G$-invariant maps between linear 
representations are now replaced by multiplicity categories, whose `dimensions' are the $\Hom$-spaces 
in the category of coherent sheaves. (In interesting examples, they are the Frobenius algebras 
underlying $2$-dimensional TQFTs.) There is a preferred family of simple 
representations, which in a sense exhausts the space of representations: they foliate 
the $BFM$ space. Every such representation is `symplectically induced' from a one-dimensional 
representation of a certain Levi subgroup of $G$: more precisely, it is the Fukaya category of 
a flag variety of $G$. This is formally similar to the Borel-Weil construction of irreducible 
representations of $G$ by holomorphic induction. Recall that in that world there is another kind 
of ``$L^2$-induction" from closed subgroups, which is right adjoint to the restriction functor. 
The counterpart of na\"ive induction also exists in our world, and gives the (curved) 
\emph{string topologies} \cite{cs} of the same flag varieties, instead of their Fukaya categories. 

This story might seem a bit unhinged, were it not for the appearance of the governing structure 
in the work of Kapustin, Rozansky and Saulina \cite{krs}. Studied there are boundary conditions in 
the $3$-dimensional TQFT associated to a holomorphic symplectic manifold $X$, known as Rozansky-Witten 
theory \cite{rozwit}. Among those are holomorphic Lagrangian sub-manifolds of $X$, or more generally, 
sheaves of categories over such sub-manifolds. (The full $2$-category of all boundary conditions 
does not yet have a precise definition.) The relation to gauge theory is summarized by the observation 
that gaugeable $2$-dimensional field theories are topological boundary conditions for pure $3$-dimensional 
topological gauge theory. The reader may illustrate this with an easy example: the representations 
of a finite group $F$ are the boundary conditions for pure $F$-gauge theory in $2$ dimensions; yet 
these representations are exactly the $1$-dimensional topological field theories (vector spaces) 
which admit $F$-symmetry. Modulo the \cite{krs} description of Rozansky-Witten theory, my entire 
story is underpinned by the following
\begin{mymeta} 
Pure topological gauge theory in $3$ dimensions for a compact Lie group $G$ is equivalent to the 
Rozansky-Witten theory for the BFM space of the Langlands dual Lie group $G^\vee$. 
\end{mymeta}

\noindent
I shall offer no elucidation of this, beyond its inspirational value; however, strong 
indications of this statement have been known in the physics literature, at 
least for special $G$ \cite{seiwit, argfar, martwarn}. Formulating this statement 
in a mathematically useable way will require an excursion through much preliminary 
material in \S2-5. A small reward will come in \S6, where we illustrate how these 
ideas can lead to `real answers'.

A closing warning is that the results in this paper are partly experimental: enough 
examples have been checked to rule out plausible alternatives, but I do not claim to know 
proofs in full generality. In fact, the status of Floer-Fukaya theory makes such claims difficult 
to sustain, and the author has no special expertise on that topic. In topological cases, 
such as for string topology (Fukaya theory of cotangent bundles), precise statements and 
proofs are possible (and easy). More generally, the results apply to the abstract setting 
of differential graded (or $A_\infty$-categories) with topological $G$-action, the question 
being to what extent the Fukaya category of a symplectic manifold with Hamiltonian $G$-action 
qualifies. (For non-compact manifolds, this depends on the `wrapping' condition at 
$\infty$.) If nothing else, the paper can be read as a template for what a nice world should 
look like. 

\subsection{Acknowledgements} I thank M.~Abouzaid, D.~Ben-Zvi, 
K.~Fukaya, K.~Hori, A.~Kapustin, A.~Neitzke, C.~Woodward for helpful comments and conversation, 
and am especially indebted to E.~Witten for explaining the relation to $4$-dimensional gauge theory and 
the Nahm equations. Many thanks are due to the Geometry group at UT Austin for the invitation to 
lecture there, where a primitive version of this material was first outlined  \cite{taus}; 
for later developments, see \cite{tlum}.

\section{Topological field theory}
Topological field theory, introduced originally by Atiyah\cite{at}, Segal \cite{seg} 
and Witten \cite{wit}, promised to systematize a slew of new $3$-manifold invariants. The 
invariants of a $3$-manifold $M$ are thought to arise from \emph{path integrals} over a space 
of maps from $M$ to a target $X$. The latter is often a manifold, but in interesting cases, 
related to \emph{gauge theory}, it is a stack. One example relevant for us will have $X$ a 
holomorphic symplectic manifold, leading to \emph{Rozansky-Witten theory} \cite{rozwit}.
The $2$-dimensional version of this notion quickly found application to the counting of 
holomorphic curves, the Gromov-Witten invariants of a symplectic manifold $X$: these 
are controlled by a family of TQFTs parametrized by the even cohomology space $H^{ev}(X)$. 

\subsection{Extended TQFTs}\label{tqftbasics}
Both theories above have a bearing on my story, once they are \emph{extended down to points}. 
In the original definition, a $d$-dimensional TQFT is a symmetric, strongly monoidal 
functor form the category whose objects are closed $(d-1)$-manifolds and whose morphisms are 
compact $d$-bordisms, to the category $\mathfrak{Vect}$ of complex finite-dimensional vector 
spaces; the monoidal structures are disjoint union and tensor product, respectively. (Some tangential 
structure on manifolds is chosen, as part of the starting datum.) Fully extending the theory 
means extending this functor to one from the \emph{bordism $d$-category} $Bord_d$, whose objects 
are points and whose $k$-morphisms are compact $k$-manifolds with corners (and some tangential structure), to 
some de-looping of the category of vector spaces: a symmetric monoidal $d$-category whose top 
three layers are complex numbers, vector spaces and linear categories, or a differential graded 
(dg) version of this. When $d=2$, 
which most concerns us, the target is usually the $2$-category $\mathfrak{LCat}$ of linear 
dg categories, linear functors and natural transformations. The reader may consult Lurie \cite{lur}, 
references therein and the wide following it inspired, for a precise setting of higher categories. 

\begin{example}[$2$-dimensional gauge theory with finite gauge group $F$]
This theory is defined for unoriented manifolds; among others, the functor $Z_F$ which sends 
a point $*$ to the category $\mathrm{Rep}(F)$ of (finite-dimensional) linear representations 
of $F$, the half-circle bordism $\subset:\emptyset\to \{*, *'\}$ to the functor $\mathfrak{Vect}
\to \mathfrak{Rep}(F)\otimes \mathfrak{Rep}(F)$ sending $\mathbb{C}$ to the ($2$-sided) regular 
representation of $F$,  the opposite bordism $\supset: \{*, *'\}\to \emptyset$ to the functor 
$\mathfrak{Rep}(F)\otimes \mathfrak{Rep}(F) \to\mathfrak{Vect}$ sending $V\otimes W$ to the 
subspace of $F$-invariants therein. 
A closed surface gives a number, which is the (weighted) count of principal $F$-bundles. 
See for instance \cite{fhlt} for a uniform construction of the complete functor and generalizations.  
\end{example}

The first theorem of \cite{lur} is that an such extended TQFT $Z:Bord_d\to 
??$ is determined by its value $Z(+)$ on the point, at least in the setting of \emph{framed} 
manifolds. The object $Z(+)$, which we call the generator of $Z$, must satisfy some 
strong (\emph{full dualizability}) conditions, but carries no additional structure, 
beyond being a member of an ambient $d$-category. 

On the other hand, the ability to pass to surfaces with less structure than a framing on their 
tangent bundle forces additional structure on the generator $Z(+)$. The point (conceived 
together with an ambient germ of surface) carries a $2$-framing, on which the group 
$\mathrm{O}(d)$ acts. Lurie's second theorem states that, given a tangential structure, encoded 
in a homomorphism $G\to \mathrm{O}(d)$, factoring the theory $Z$ from $Bord_d$ through the category 
$Bord_d^{\,G}$ of $d$-folds with $G$-structure is equivalent to exhibiting $Z(+)$ as a fixed-point 
for the $G$-action on the image of TQFTs in the target$d$- category (more precisely, the 
sub-groupoid of fully dualizable objects and invertible morphisms). 

The best-known case of oriented surfaces, when $G=\mathrm{SO}(2)$, requires a \emph{Calabi-Yau} 
structure on $Z(+)$. This can be variously phrased: as a trivialization of the \emph
{Serre functor}, which is an automorphism of any fully dualizable linear dg category (see 
Remark~\ref{serre} below); alternatively, as a linear functional on the cyclic 
homology of $Z(+)$ whose restriction to Hochschild homology $HH_*(Z(+))$ induces 
a perfect pairing on $\Hom$ spaces:
\[
\Hom(x,y)\otimes \Hom(y,x) \to \Hom(x,x) \to HH_* \to \mathbb{C}. 
\]
This case of Lurie's theorem recovers earlier results of Costello, Kontsevich and Hopkins-Lurie 
\cite{cos, ks}.   

The Hochschild homology $HH_*(Z(+))$ is meaningful in a different guise: it is the space of 
states $Z(S^1)$ of the theory, for the circle with the radial framing. The circle is pictured 
here with a germ of surrounding surface, and therefore carries a $\mathbb{Z}$'s worth of 
framings, detected by a winding number. The Hochschild cohomology $HH^*$ goes with the 
blackboard framing, and the space for the framing with winding number $n$ is $HH^*$ of the 
$n$th power of the Serre functor. (Of course, for oriented theories there is no framing 
dependence, and these spaces agree.) 

\subsection{Topological group actions}
An important point is that the action of $\mathrm{O}(2)$ (and thus $G$) on the target 
category $Z(+)$ is \emph{topological}, or factored through its topology. There are several 
ways to formulate this constraint, which is vacuous when $G$ is discrete. The favored 
formulation will depend on the nature of the target category; in the linear case, and when 
$G$ is connected, we will provisionally settle for the one in Theorem~\ref{e2action} below. 
Combined with Statement~\ref{omegaG} below, this generalizes an old result of Seidel~\cite{seid} 
on Hamiltonian diffeomorphism groups. 

Here are some alternative definitions:

\begin{enumerate}
\item We can ask for a \emph{local trivialization} of the action in a contractible neighborhood 
of $1\in G$, an isomorphism with the trivial action of that same neighborhood (up to coherent 
homotopies of all orders).
\item Using the action to form a bundle of categories with fiber $Z(+)$ over the classifying 
stack $BG$, we ask for an integrable flat connection on the resulting bundle of categories. 
(Formulating the flatness condition requires some care, in light of the fiber-wise 
automorphisms.)

\item Exploiting the contractibility of the group $P_1G$ of paths starting at $1\in G$, we can 
ask for a trivialization of the lifted $P_1G$-action. 

Now, the action of the based loop group $\Omega G$ (kernel of $P_1G\to G$) is already trivial 
(being factored through $1\in G$), and the difference of trivializations defines a (topological) 
representation of $\Omega G$ by automorphisms of the identity functor in $Z(+)$. 

The group $\Omega G$ has an $E_2$ structure,  seen from its equivalence with the second 
loop space $\Omega^2BG$; and the representation on $\mathrm{Id}_{Z(+)}$ is the $2$-holonomy, 
over spheres, of the flat connection in \#2. Importantly, it is an $E_2$ representation.
\end{enumerate}

\begin{remark}
When $G$ is connected, description \#3 above captures all the information for the action 
(up to contractible choices), because the space of trivializations of a trivial topological 
action of $P_1G$ is contractible.
\end{remark}

\begin{example}\label{serre}
A topological action of the circle on a category is given by a group homomorphism 
from $\mathbb{Z}=\pi_1S^1 = \pi_0\Omega S^1$ to the automorphisms of the identity: 
equivalently, a central (in the category) automorphism of each object. Because there is no 
higher topology in $S^1$, this also works when the target is a $2$-category, such as
the (sub-groupoid of fully dualizable objects in the) $2$-category $\mathfrak{LCat}$, 
the structural $\mathrm{SO}(2)\subset\mathrm{O}(2)$ action gives an automorphism 
of each category: this is the Serre functor. 
\end{example}

\begin{example}\label{cuspidalreps}
Endomorphisms of the identity in the linear category $\mathfrak{Vect}$ are the complex scalars, 
so that linear topological representations of a connected $G$ on $\mathfrak{Vect}$ are 
$1$-dimensional representations of $\pi_0\Omega G\cong \pi_1G$. These are the points in 
the center of the complexified Langlands dual group $G^\vee_{\mathbb{C}}$.
\end{example}

Recall that the endomorphisms of the identity in a category (the center) form the 
$0^{\mathrm{th}}$ Hochschild cohomology. To generalize the above example to the 
derived world, we should include the entire Hochschild cochain complex.  

\begin{theorem}\label{e2action}
Topological actions of a connected group $G$ on a linear dg-category $\mathfrak{C}$ are captured 
(up to contractible choices) by the induced $E_2$ algebra homomorphism from the chains $C_*\Omega G$, 
with Pontrjagin product, to the Hochschild cochains of $\mathfrak{C}$. \qed
\end{theorem} 

\begin{example}
From a continuous action of $G$ on a space $X$, we get a locally trivial action on the 
cochains $C^*X$. Indeed, we get an action of $\Omega G$ on the free loop space $LX$ of $X$. 
The action is fiber-wise with respect to the bundle $\Omega X \to LX \to X$. Let $C^*\big(
X;C_*\tilde{\Omega}X\big)$ be the cochain complex on $X$ with coefficients in the fiber-wise 
chains for this bundle. With the fiber-wise Pontrjagin product, this is a model for the 
Hochschild cochains of the algebra $C^*(X)$, and the action of $\Omega G$ exhibits the 
$E_2$ homomorphism in the theorem. 
\end{example}

\begin{remark}
The ``$E_2$'' in the statement is not jus a commutativity constraint, but can contain 
(infinite amounts of!) data; see Lesson \ref{lessons}.5.
\end{remark}

\begin{remark}
One floor below, for $1$-dimensional field theories, the category $Z(+)$ is replaced with 
a vector space (or a complex), and we recognize \#2 above as defining a topological 
representation of $G$. The datum in Theorem~\ref{e2action}  is replaced by an ($E_1)$ 
algebra homomorphism from the chains $C_*G$, with Pontrjagin product, to 
$\mathrm{End}(Z(+))$; there is no connectivity assumption. Climbing to the higher ground 
of $n$-categories, we can extract an $E_{n+1}$-algebra homomorphism from $C_*\Omega^nG$ to the $E_n$ 
Hochschild cohomology; but this misses the information from the homotopy of $G$ below $n$.  
\end{remark}

The following key example captures the relevance of my story to real mathematics. (In fact, 
it contains \emph{all} examples I know for topological group actions!) 

\begin{conjecture}\label{omegaG}
Let $G$ act in Hamiltonian fashion action on a symplectic manifold $X$. Then, $G$ acts 
topologically on the Fukaya category of $X$.
\end{conjecture}

\begin{proof}
A Hamiltonian action of $G$ on $X$ defines, in the category of symplectic manifolds 
and Lagrangian correspondences, an action of the group object $T^*G$.\footnote{The moment 
map $\mu:X\to \mathfrak{g}^*$ appears in the requisite Lagrangian, 
$\{(g,\mu(gx)_, x, gx)\} \subset T^*G\times (-X)\times X$.} This makes the 
Fukaya category of $X$ into a module category over the wrapped Fukaya category 
$\mathfrak{WF}(T^*G)$. A theorem of Abouzaid \cite{ab} identifies the latter with 
that of $C_*\Omega G$-modules. The tensor structure is identified with the $E_2$ 
structure of the Pontrjagin product, by detecting it on generators of the category 
(the cotangent fibers). 
The resulting structure is equivalent to the datum in Theorem~\ref{e2action}.
\end{proof}

\begin{remark}
It may seem strange to state a conjecture and then provide a proof. However, the reader 
will detect certain assumptions which have not been clearly stated in the conjecture: 
mainly, functoriality of Fukaya categories under Lagrangian correspondences. If $X$
is non-compact, equivariance of the wrapping condition at $\infty$ is essential; 
the statement fails for the \emph{infinitesimally wrapped} Fukaya category of Nadler 
and Zaslow \cite{nz}, see below. (Another outline argument is more tightly connected to 
holomorphic disks and $G_{\mathbb{C}}$-bundles, but that relies on details of 
the construction of the Fukaya category.) 
\end{remark}
  
\begin{remark}
A closely related notion to the one discussed, but distinct from it, is that of an \emph
{infinitesimally trivialized} Lie group action. Here, we ask for the action to be 
differentiable, and the restricted action to the formal group $\hat{G}$ (equivalently, 
the Lie algebra $\mathfrak{g}$) should be homologically trivialized. An example is 
furnished by an action of $G$ on a manifold $X$ and the induced action on the algebra 
$\mathcal{D}(X)$ of differential operators: the Lie action of $\mathfrak{g}$ is 
trivialized in the sense that it is inner, realized by the natural Lie homomorphism 
from $\mathfrak{g}$ to the $1$st order differential operators. Theorem~\ref{e2action} 
does \emph{not} usually apply to such situations. With respect to the alternative definition
\#2 above, the relevant distinction is between \emph{flat} and \emph{integrable} 
connections over $BG$. 
\end{remark}

\subsection{Gauging a topological theory}   
Given a guantum field theory and a (compact Lie) group $G$, physicists normally produce 
a $G$-gauged theory in two stages. The theory is first coupled to a `classical gauge 
background', a principal $G$-bundle. (No connection is needed in the case of topological 
actions.\footnote{Flat connections would be needed when $G$ action does not factor through 
topology, as in $B$-model theories.}) Then, we `integrate over all principal bundles' 
to quantize the gauge theory.

These two distinct stages are neatly spelt out in the setting of extended TQFTs. 
Lurie's theory already captures the first stage of gauging. Namely, 
we convert the principal $G$-bundle into a tangential structure by choosing the 
trivial homomorphism $G\to \mathrm{O}(2)$. (Of course, we may add any desired 
tangential structure, such as orientability, by switching to $G\times \mathrm{SO}(2)\to 
\mathrm{O}(2)$, by projection.) Making $Z(+)$ into a fixed point for the trivial 
$G$-action means defining a (topological) $G$-action on $Z(+)$. This is the input 
datum for a classically gauged theory.

Quantizing the gauge theory, or integrating over principal $G$-bundles, is tricky. 
It is straightforward for finite groups: integration of numbers is a weighted sum, 
and integration of vector spaces and categories is a finite limit or colimit. 
(The duality constraints require the limits and colimts to agree; working in 
characteristic $0$ ensures that \cite{fhlt}.) For Lie groups $G$, integration of 
the numbers requires a fundamental class on the moduli of principal bundles. 
For instance, the symplectic volume form is relevant to topological Yang-mills theory. 
A limited $K$-theoretic fundamental class was defined in \cite{tw}, and cohomological 
classes, such as the one relevant to topological Yang-Mills theory, can be extracted 
from it. But this matter seems worthy of more subtle discussion than space allows here. 

In fact, the gauge theory \emph{cannot} always be fully quantized. The generating 
object for the quantum gauge theory is the invariant category $Z(+)^G$, 
which agrees with the co-invariant category $Z(+)_G$ under mild assumptions. 
In the framework of Theorem~\ref{e2action}, we compute the generator $Z(+)_G$ as 
a tensor product
\begin{equation}\label{gaugedtheory}
Z(+)_G = Z(+) \otimes _{C_*\Omega G} \mathfrak{Vect}
\end{equation} 
with the trivial representation. The $1$-dimensional part of the field theory, 
and sometimes part of the surface operations, are well-defined; but the complete 
surface-level operations often fail to be defined. Thus, for the trivial $2D$ theory, $Z(+)=
\mathrm{dg-}\mathfrak{Vect}$ with trivial $G$-action, and the fixed-points are 
local systems over $BG$. This generates a partially defined $2D$ theory, a 
version of string topology for the space $BG$. The space associated to the 
circle is the equivariant cohomology $H^*_G(G)$ for the conjugation action, 
and the theory is defined the subcategory of $Bord_2$ where all surfaces 
(top morphisms) have non-empty output boundaries for each component. 

This example can be made more interesting by noting that the trivial action of $G$ 
on $\mathrm{dg-}\mathfrak{Vect}$ has interesting topological deformations, in the 
$\mathbb{Z}/2$-graded world; the notable one comes from the quadratic Casimir in 
$H^4(BG)$, and gives topological Yang-Mills theory with gauge group $G$. 
When $G$ is semi-simple, this theory is almost completely defined, and the invariants of 
a closed surface (of genus $2$ or more) are the symplectic volumes of the moduli spaces of 
flat connections. (Further deformations exits, by the entire even cohomology of $BG$ and 
relate to more general integrals over those spaces.) These should be regarded as 
twisted Gromov-Witten theories with target space $BG$. A starting point of the 
present work was the abject failure of the homological calculation~\eqref{gaugedtheory} 
in these examples: for topological Yang-Mills theory, \eqref{gaugedtheory} gives 
the zero answer when $G$ is simple.

\subsection{The space of states} Independently of good behavior of the fixed-point 
category $Z(+)^G$, the space(s) of states of the gauged theory are well-defined. More 
precisely, each $g\in G$ gives an autofunctor $g_*$ of the category. The Hochshild 
cochain complexes $HCH^*(g_*;Z(+))$ assemble to a (derived) local system $\mathcal{H}(Z(+))$ 
over the group $G$, which is equivariant for the conjugation action, and the space of states 
for the (blackboard framed) circle in the gauge theory is the equivariant homology 
$H^G_*(G;\mathcal{H})$. It has a natural $E_2$ multiplication, using the Pontrjagin product 
in the group. When $Z(+)=\mathfrak{Vect}$, with the trivial $G$-action, we recover the 
string topology space $H^G_*(G)$ of $BG$ by exploiting Poincar\'e duality on $G$.\footnote
{The last space goes with the radially framed circle.}   

\section{The $2$-category of Kapustin-Rozansky-Saulina }
As the image of the point, an object in the $3$-dimensional bordism $3$-category, Lurie's 
generator for pure $3$-dimensional gauge theory should have categorical depth $2$.
My proposal for this generator is a $2$-category associated to a certain holomorphic 
symplectic manifold, to be described in \S\ref{bfmsect}. 

Fortunately, the existence of the requisite $2$-category has already been conjectured, and a 
proposal for its construction has been outlined in  \cite{krs, kr}. When 
$X$ is compact, this $2$-category should generate the Rozansky-Witten 
theory \cite{rozwit} of $X$. In particular, its Hochschild cohomology, which on general grounds 
is a $1$-category with a braided tensor structure, should be (a dg refinement of) the derived 
category of coherent sheaves on $X$ described in \cite{robwil}. Just like Rozansky-Witten theory, 
the narrative takes place 
in a differential graded world, and in applications, the integer grading must be collapsed$\mod 2$
(the symplectic form needs to have degree $2$, if the integral grading is to be kept). To 
keep the language simple, I will use `sheaf' for `complex of sheaves' and write $\mathfrak{Coh}$
for a differential graded version of the category of coherent sheaves, etc.  

\begin{remark}
The $2$-category may at first  appear analogous to the deformation quantization of the 
symplectic manifold; but that is not so. That analogue --- a double categorification --- is 
$\mathfrak{Coh}(X)$ with its braided tensor structure. The category \cite{krs} is a 
`square root' of that, and I will denote it $\sqrt{\mathfrak{Coh}}(X)$ or $KRS(X)$.  
\end{remark}

\subsection{Simplified description} The following partial description of the $KRS$ 
$2$-category applies to a Stein manifold $X$, when deformations  coming from coherent cohomology 
vanish.\footnote{My discussion is faulty in another way, failing to incorporate the Spin structures, 
which must be carried by the Lagrangians. I am grateful to D.\ Joyce for flagging their role.}  
In our example, $X$ will be affine algebraic.
Among objects of $\sqrt{\mathfrak{Coh}}(X)$ are smooth holomorphic Lagrangians $L\subset X$; 
more general objects are coherent sheaves of $\mathcal{O}_L$-linear categories on such $L$. 
(The object $L$ itself stands for its dg category $\mathfrak{Coh}(L)$ of coherent sheaves, 
a generator for the above.) To make this even more precisee, $\sqrt{\mathfrak{Coh}}(X)$ is the sheaf of global 
sections of a coherent sheaf of $\mathcal{O}_X$-linear $2$-categories, whose localization at 
any smooth $L$ as above is equivalent the $2$-category of module categories over the sheaf of 
tensor categories $(\mathfrak{Coh}(L),\otimes)$ on $L$; with a bit of faith, this  
pins down $\sqrt{\mathfrak{Coh}}(X)$, as follows.

For two Lagrangians $L,L'\in X$, $\Hom(L,L')$ will be a sheaf of categories supported on $L\cap L'$, 
and a $(\mathfrak{Coh}(L),\otimes)-(\mathfrak{Coh}(L'),\otimes)$ bi-module. Localizing at $L$, 
we choose a (formal) neighborhood identified symplectically with $T^*L$, so that we regard 
(locally) $L'$ as the graph of a differential $d\Psi$, for a \emph{potential} function $\Psi:
L\to \mathbb{C}$. Locally where this identification is valid, $\Hom(L,L')$ becomes equivalent 
to the \emph{matrix factorization} category $MF(L, \Psi)$. (See for instance \cite{orl}.)  

\subsection{Lessons}\label{lessons} Several insights emerge from this important notion.
\begin{enumerate}
\item A familiar actor in mirror symmetry, a complex manifold $L$ with potential $\Psi$,  
is really the object in $\sqrt{\mathfrak{Coh}}(T^*L)$ represented by the graph $\Gamma(d\Psi)$, 
masquerading as a more traditional geometric object. 
The matrix factorization category $MF(L,\Psi)$ is its $\Hom$ with the zero-section.
This resolves the contradiction in which the restriction of the category $MF(L,\Psi)$ to 
a sub-manifold $M\subset L$ is commonly taken to be the matrix factorization category of $\Psi|_M$. 
That is clearly false in the $2$-category of $(\mathfrak{Coh}(L),\otimes)$-module categories
(the result of localizing to the zero-section $L\subset T^*L$). For instance, if the critical 
locus of $\Psi$ does not meet $M$, $\Hom$ computed in $(\mathfrak{Coh}(L),\otimes)$-modules 
gives zero. Instead, $M$ must be replaced by the object represented by its co-normal bundle 
in $\sqrt{\mathfrak{Coh}}(T^*L)$, whose $\Hom$ there with $\Gamma(d\Psi)$ computes precisely $MF(M,\Psi|_M)$. 
 
\item The well-defined assignment sends $(\mathfrak{Coh}(L),\otimes)$-module categories 
to sheaves of categories with Lagrangian support in the cotangent bundle $\hat{T}^*L$, completed 
at the zero-section. Namely, the Hochschild cohomology of such a category $\mathfrak{K}$ is 
(locally on $L$) an $E_2$-algebra over the second ($E_2$) Hochschild cohomology of 
$(\mathfrak{Coh}(L),\otimes)$, which is an $E_3$ algebra. The spectrum of the latter is 
$\hat{T}^*L$, with $E_3$ structure given by the standard symplectic form. This turns 
$\mathrm{Spec}\,HH^*(\mathfrak{K})$ into a coherent sheaf with co-isotropic support in 
$\hat{T}^*L$, and $\mathfrak{K}$ sheafifies over it. The \emph{Lagrangian} condition is 
clearly related to a finiteness constraint, but this certainly shows the need to include 
singular Lagrangians in the $KRS$ $2$-category. 

\item The deformation of a $(\mathfrak{Coh}(L),\otimes)$-module category $\mathfrak{M}$ 
by the addition of a potential (`curving') $\Psi\in \mathcal{O}(L)$ shifts the support 
of $\mathfrak{M}$ vertically by $d\Psi$ in $T^*L$. This allows one to move from formal to analytic 
neighborhoods of $L$, if the deformation theory under curvings is well-understood.
For instance, one can compute the $\Hom$ between two objects that do not intersect 
the zero-section --- such as two potentials without critical points --- by drawing 
their intersection into $L$: $\Hom(\Gamma(d\Phi),\Gamma(d\Psi)) = MF(L,\Psi-\Phi)$. 

\item More generally, Hamiltonian vector fields on $\hat{T}^*L$ give the derivations 
of $\sqrt{\mathfrak{Coh}}(\hat{T}^*L)$ defined from its $E_2$ Hochschild cohomology. Hamiltonians vanishing 
on the zero-section preserve the latter, and give first-order automorphisms of 
$(\mathfrak{Coh}(L),\otimes)$. 

\item The $KRS$ picture captures in geometric terms sophisticated algebraic information. 
For example, the category $\mathfrak{Vect}$ can be given a $(\mathfrak{Coh}(L),\otimes)$-module 
structure in many more ways in the $\mathbb{Z}/2$ graded world: any potential $\Psi$ with 
a single, Morse critical point will accomplish that. The location of the critical 
point $p\in L$ misses an infinite amount of information, which is captured precisely by 
the graph of $d\Psi$; this is equivalent to an $E_2$ structure on the evaluation 
homomorphism $\mathcal{O}_L\to \mathbb{C}_p$ at the residue field (cf.~Theorem~\ref{e2action}).  
\end{enumerate}    

\noindent Parts of this story can be made rigorous at the level of formal deformation theory, 
see for instance \cite{fran}, and of course the outline in \cite{kr}. Lesson 3 also offers a 
working definition of the $2$-category $\sqrt{\mathfrak{Coh}}(T^*L)$ as 
that of $(\mathfrak{Coh}(L),\otimes)$-modules, together with all their deformations 
by curvings. On a general symplectic manifold $X$, we can hope to patch the local definitions 
from here.\footnote{If $X$ is not Stein, deformations will be imposed upon this story  
by coherent cohomology.} It is not my purpose to supply a construction of $\sqrt{\mathfrak{Coh}}(X)$ 
here --- indeed, that is an important open question --- but rather, to indicate enough structure 
to explain my answer to the mirror of (non-abelian) gauge theory. I believe that one important 
reason why that particular question has been troublesome is that the mirror holomorphic symplectic 
manifold, the $BFM$ space of \S\ref{bfmsect}, \emph{not} quite a cotangent bundle, so the usual 
description in terms of complex manifolds with potentials is inadequate.

\begin{remark}
If $X=T^*L$ for a manifold $L$, and we insist on integer, rather than $\mathbb{Z}/2$-gradings, 
then the cotangent fibers have degree $2$ and all structure in the $KRS$ category is invariant 
under the scaling action on $T^*L$. In that case, we are dealing precisely with 
$(\mathfrak{Coh}(L),\otimes)$-modules.  
\end{remark}

\subsection{Boundary conditions and domain walls} \label{boundarycond}
The $\Hom$ category $\Hom(L,L')$ for two Lagrangians $L,L'\subset X$ with finite intersection supplies 
a $2$-dimensional topological field theory for framed surfaces; this follows form its local
description by  matrix factorizations. Since $X$ itself aims to 
define a $3D$ (Rozansky-Witten) theory and each of $L,L'$ is a boundary condition for it, 
one should picture a sandwich of Rozansky-Witten filling between a bottom slice of $L$ and a 
top one of $L'$. The formal description is that $L,L': \mathrm{Id}\to RW_X$ are morphisms from 
the trivial $3D$ theory $\mathrm{Id}$ to Rozansky-Witten theory $RW_X$, viewed as functors 
from $Bord_2$ to the $3$-category of linear $2$-categories, and the category $\Hom(L,L')$ of 
natural transformations between these morphisms is the generator for this sandwich theory. 
Geometrically, it is represented by the interval, with $RW_X$ in the bulk and $L,L'$ at the ends, 
and is also known as the \emph{compactification} of $RW_X$ along the interval, with the named boundary conditions.   

Factoring this theory through \emph{oriented} surfaces requires a trace on the Hochschild 
homology $HH_*$ (cf.~\S\ref{tqftbasics}). Now, the canonical description of the only non-zero 
group, $HH_{\dim L}$,  turns out to involve the Spin square roots\footnote{The \emph{cohomology} 
is easy to pin down canonically, as the functions on $L\cap L'$.} of the 
canonical bundles $\omega, \omega'$  of $L,L'$ on their scheme-theoretic overlap:
\begin{equation}\label{ext}
HH_{\dim L} \Hom(L,L') \cong \Gamma\big(L\cap L'; (\omega\otimes \omega')^{1/2}\big).
\end{equation}
A non-degenerate quadratic form on $HH_{\dim L}$ comes from the Grothendieck residue 
(and the symplectic volume on $X$). A non-degenerate trace on $HH_*$ will thus be defined 
by choosing non-vanishing sections of $\omega^{1/2}, \omega'^{1/2}$ on $L,L'$.

\begin{remark}
A generalization of the notion of boundary condition is that of a \emph{domain wall} 
between TQFTs. This is an adjoint pair of functors between the TQFTs meeting certain 
(dualizability) conditions, see \cite{lur},~\S4. A boundary condition is a domain wall with 
the trivial TQFT. Just as a holomorphic Lagrangian in $X$ can be expected to define a boundary
condition for $RW_X$, a holomorphic Lagrangian correspondence $X \leftarrow C \rightarrow Y$ 
should define a domain wall between $RW_X$ and $RW_Y$. We shalll use these in \S5 and \S6, in  
comparing gauge theories for different groups.
\end{remark}

\section{The mirror of abelian gauge theory}
This interlude recalls the mirror story of torus gauge theory; except for the difficulty 
mentioned in Lesson~1 of \S\ref{lessons}, this story is well understood and can be phrased
as a categorified  Fourier-Mukai transform. In fact, in this case 
we can indicate the other mirror transformation, from the gauged $B$-model to a family 
of $A$-models.

\subsection{The $\mathbb{Z}$-graded story} We will need to correct this when abandoning 
$\mathbb{Z}$-gradings, in light of the wisdom of the previous section; nevertheless 
the following picture is nearly right.

\begin{proposition}\label{torusreps}
(i) Topological actions of the torus $T$ on the category $\mathfrak{Vect}$ are classified by points 
in the complexified dual torus $T^\vee_{\mathbb{C}}$.\\
(ii) A topological action of $T$ on a linear category $\mathfrak{C}$ is equivalent to a 
quasi-coherent sheafification of $\mathfrak{C}$ over $T^\vee_{\mathbb{C}}$.
\end{proposition} 
\begin{proof}
Both statements follow from Theorem~\ref{e2action}, considering that the group ring $C_*(\Omega T)$ is 
quasi-isomorphic to the ring of algebraic functions on $T^\vee_{\mathbb{C}}$, and that a category 
naturally sheafifies over its center, the zeroth Hochschild cohomology.
\end{proof}

There emerges the following $0^{\mathrm{th}}$ order approximation to abelian gauged mirror symmetry: 
if $X$ is a symplectic manifold with Hamiltonian action of $T$, and $X^\vee$ is a mirror of $X$ --- 
in the sense that $\mathfrak{Coh}(X^\vee)$ is equivalent to the Fukaya category $\mathfrak{F}(X)$ --- 
then the group action on $X$ is mirrored into a holomorphic map $\pi: X^\vee\to T^\vee_{\mathbb{C}}$. 
This picture could be readily extracted from Seidel's result, \cite{seid}.

Proposition~\ref{torusreps} interprets the mirror map $X^\vee\to T^\vee_{\mathbb{C}}$ 
as a \emph{spectral decomposition} of the category $\mathfrak{F}(X)$ into irreducibles 
$\mathfrak{Vect}_\tau$. One of the motivating conjectures of this program gives a geometric 
interpretation of this spectral decomposition, in terms of the original manifold $X$ and 
the moment map $\mu: X\to \mathfrak{t}^*$.

\begin{conjecture}[Torus symplectic quotients]\label{torusgit}
The multiplicity of $\mathfrak{Vect}_\tau$ in $\mathfrak{F}(X)$ is the Fukaya category of the 
symplectic reduction of $X$ at the point $\mathrm{Re}\log\tau\in \mathfrak{t}^*$, with imaginary 
curving ($B$-field) $\mathrm{Im}\log\tau$.
\end{conjecture}

\begin{remark}
This is, for now, meaningless over singular values of the moment map, where there seems to be no 
candidate definition for the Fukaya category of the quotient.
\end{remark}

\begin{remark}
The conjecture relies on using the \emph{unitary mirror} of $X$, constructed from Lagrangians with 
unitary local systems. Otherwise, in the toric case, the algebraic mirror $X^\vee$ is $T^\vee_{
\mathbb{C}}$, obviously having a point fiber over every point in $T^\vee_{\mathbb{C}}$; yet the 
symplectic reduction is empty for values outside the moment polytope. 
That polytope is precisely the cut-off prescribed for the mirror by  unitarity. 
\end{remark}

\begin{example}[Toric varieties]\label{toricvar}
The following construction of mirrors for toric manifolds, going back to the work of Givental and 
Hori-Vafa, illustrates both the conjecture and the need to correct the picture by moving to the 
$KRS$ category. 

Start with the mirror of $X=\mathbb{C}^N$, with standard symplectic form, as the space 
$T^\vee_{\mathbb{C}}:=(\mathbb{C}^\times)^N$ with potential $\Psi = z_1+\dots z_N$. 
Here, $T^\vee$ is the dual of the diagonal torus acting on $X$, and the mirror map 
$X^\vee\to  T^\vee_{\mathbb{C}}$ is the identity.\footnote{This is readily obtained from the 
SYZ picture, using coordinate tori as Lagrangians; the unitary mirror is cut off by $|z_k|< 1$.} 
For a sub-torus $i:K\hookrightarrow T$, the mirror of the symplectic reduction 
$X_q:=\mathbb{C}^N/\!/\!_qK$ at $q\in \mathfrak{k}^*$ is the (torus) fiber $X^\vee_q$ of 
the dual surjection $i^\vee: T^\vee_{\mathbb{C}} \twoheadrightarrow K^\vee_{\mathbb{C}}$, 
with restricted super-potential $\Psi$. The parameter $q$ lives in the small quantum cohomology of $X$.
We see here the familiar, but faulty restriction to the fiber of the matrix factorization category 
$MF(T^\vee_{\mathbb{C}},\Psi)$ of Lesson~\ref{lessons},~\#1. The problem is glaring, 
because the original MF category is null. 

The mirror $X^\vee_q$ projects isomorphically to the kernel $S^\vee_{\mathbb{C}}$ of $i^\vee$; 
this is the map $\pi$ mirror to the action of $S = T/K$ on $X$. 
\end{example}

\subsection{Fourier transform}
As can be expected in the abelian case, the spectral decomposition of Proposition~\ref{torusreps} 
is formally given by a Fourier transform. Specifically, there is a `categorical Poincar\'e 
line bundle' 
\[
\mathfrak{P}\to BT_{\mathbb{C}}\times T^\vee_{\mathbb{C}},
\] 
with an integrable flat connection along $BT$. (Of course, $\mathfrak{P}$ is the universal  
one-dimensional topological representation of $T$, and its fiber over $\tau\in T^\vee_{\mathbb{C}}$ 
is $\mathfrak{Vect}_\tau$.) Given a category $\mathfrak{C}$ with topological $T$-action, 
we form the bundle $\Hom(\mathfrak{P},\mathfrak{C})$ and integrate along $BT_{\mathbb{C}}$ 
to obtain the spectral decomposition of $\mathfrak{C}$ laid out over $T^\vee_{\mathbb{C}}$. 

\begin{remark}[$B$ to $A$]
The interest in this observation stems from a related Fourier transformation, giving a 
``$B$ to $A$'' mirror symmetry. There is another Poincar\'e bundle $\mathfrak{Q}\to 
BT_{\mathbb{C}}\times T^\vee_{\mathbb{C}}$, with flat structure this time along $T^\vee_
{\mathbb{C}}$. It may help to exploit flatness and descend to $B(T_{\mathbb{C}}\times 
\pi_1(T)^\vee)$, in which case $\mathfrak{Q}$ is the line $\mathfrak{Vect}$ with action 
of the group $T\times\pi_1(T)^\vee$, defined by the Heisenberg $\mathbb{C}^\times$-central 
extension. (The extension is a multiplicative assignment of a line to every group element, 
and the action on $\mathfrak{Vect}$ tensors by that line.) 

Fourier transform converts a category $\mathfrak{C}$ with (non-topological!) $T$-action into 
a local system $\tilde{\mathfrak{C}}$ of categories over $T^\vee_{\mathbb{C}}$. The fiber of 
$\tilde{\mathfrak{C}}$ over $1$ is the fixed-point category $\mathfrak{C}^T$, and the monodromy 
action of $\pi_1(T^\vee)$ comes from the natural action thereon of the category 
$\mathfrak{Rep}(T)$ of complex $T$-representations. For example, when $\mathfrak{C}
=\mathfrak{Coh}(X)$, the (dg) category of coherent sheaves on a complex manifold with 
holomorphic $T$-action, $\mathfrak{C}^T$ is, almost by definition, the category of sheaves 
on the quotient stack $X/T_{\mathbb{C}}$. The analogue of Conjecture~\ref{torusgit} 
is completely obvious here.

I do not know a non-abelian analogue of this ``$B$ to $A$'' story. 
\end{remark}

\subsection{The $\mathbb{Z}/2$-graded story} \label{torusz2} In light of Lesson~\ref{lessons}.1  
and Example~\ref{toricvar}, the only change needed to reach the true story is to replace the 
$(\mathfrak{Coh}(T^\vee_{\mathbb{C}}),\otimes)$-module category $\mathfrak{Coh}(X^\vee)$, 
determined from  $\pi:X^\vee\to T^\vee_{\mathbb{C}}$, by an object in the $KRS$ category 
of $T^*T^\vee_{\mathbb{C}}$: the category with $T$-action sees precisely the germ 
of a $KRS$ object near the zero-section.

This enhancement of information relies upon knowing not just the Fukaya category $\mathfrak{F}(X)$ 
with its torus action, but all of its curvings with respect to functions lifted from the mirror 
map $\pi:X^\vee\to T^\vee_{\mathbb{C}}$. However, we can expect in examples that a meaningful 
geometric construction of the mirror would carry that information. For instance, in Example~\ref
{toricvar}, we replace $(X_q^\vee, \Psi|_{X^\vee_q})$ and its map to $S^\vee_{\mathbb{C}}$ by 
the graph of $d\Psi|_{X^\vee_q}$ in $T^*S^\vee_{\mathbb{C}}$; this is the result of intersecting 
the graph of $d\Psi$ with the cotangent space at $q\in K^\vee_{\mathbb{C}}$. 

Figure~1 attempts to capture the distinction between 
$(\mathfrak{Coh}T^\vee_{\mathbb{C}},\otimes)$-modules and their $KRS$ enhancement. 
The squiggly line stands for (the support of) a general object; its germ at the zero-section is the 
underlying category, with topological $T$-action. In that sense, the zero-section represents 
the regular representation of $T$ (its $\Hom$ category with any object recovers the underlying category.) 
The invariant category is the intercept with the trivial representation, the cotangent 
space at $1\in T^\vee_{\mathbb{C}}$; other spectral components are intercepts with vertical axes. 
We see that the invariant subcategory is computed `far' from the underlying category, and a homological 
calculation centered at the zero-section will fail.

\begin{figure}
\includegraphics[height=3in]{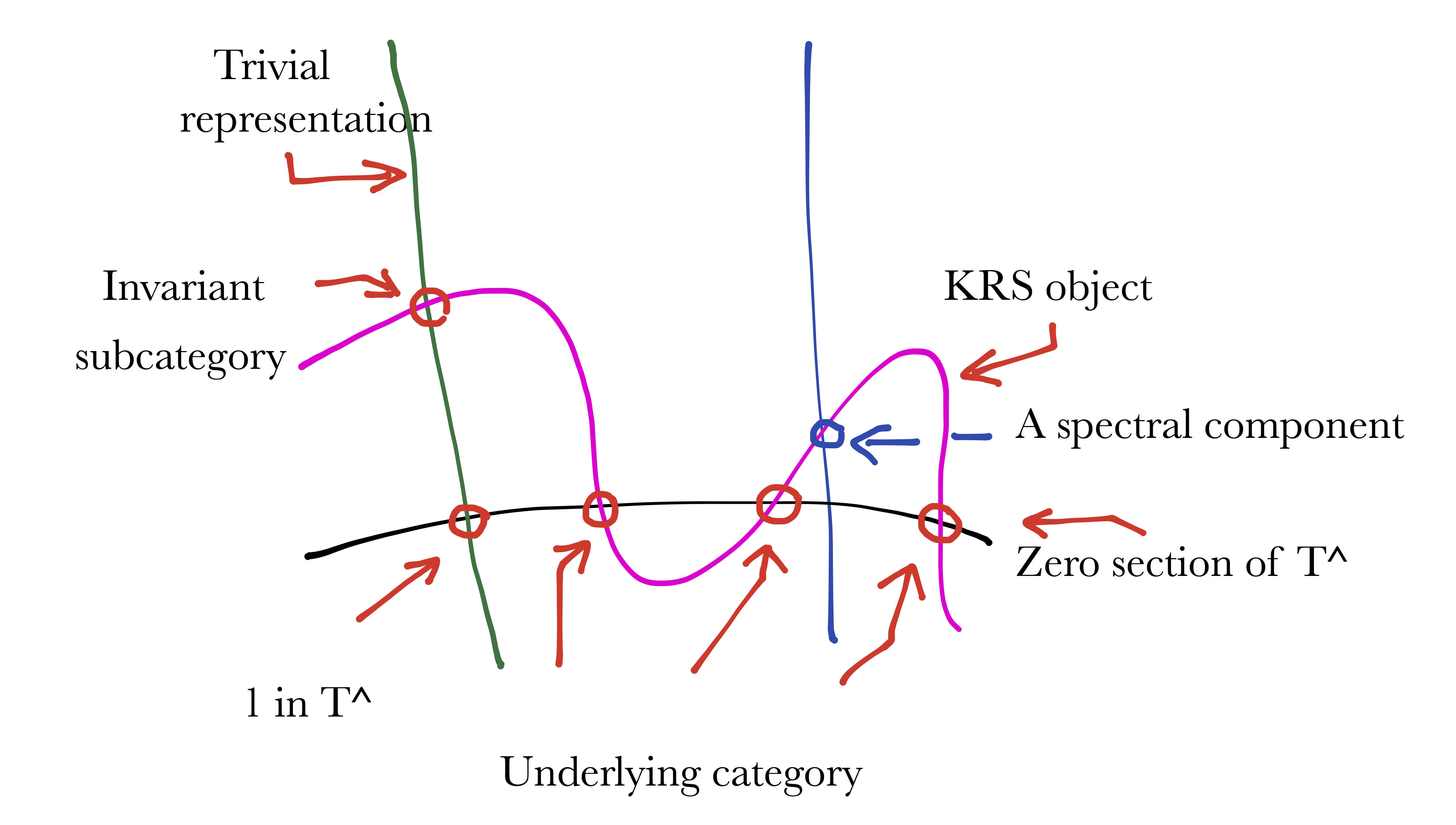}
\caption{Pictorial representation of $\sqrt{\mathfrak{Coh}}(T^*T^\vee_{\mathbb{C}})$}
\end{figure}

\section{The non-abelian mirror $BFM(G^\vee)$}\label{bfmsect}
For torus actions, the insight was that gauging a Fukaya category $\mathfrak{F}(X)$ amounted 
to enriching it from a $\mathfrak{Coh}(T^\vee_{\mathbb{C}})$ module to an object in $\sqrt
{\mathfrak{Coh}}(T^*T^\vee_{\mathbb{C}})$. In a cotangent bundle, this promotion may seem 
modest. A non-abelian Lie group $G$ will move us to a more sophisticated holomorphic 
algebraic manifold which is \emph{not} a cotangent bundle. Let $T$ be a maximal torus of $G$, 
$W$ the Weyl group and $B, B_+$ two opposite (lower and upper triangular) Borel subgroups, 
$N,N_+$ their unipotent radicals; Fraktur letters will stand for the Lie algebras and
$^\vee$ will indicate their counterparts in the Langlands dual Lie group $G^\vee$.

\subsection{The home of $2D$ gauge theory} The space $BFM(G)$ was introduced and studied by 
Bezrukavnikov, Mirkovic and Finkelberg \cite{bfm} in general, but special instances were known 
in many guises. 
Here are several descriptions. Call $T^*_{\mathrm{reg}} G_{\mathbb{C}} \subset T^* G_{\mathbb{C}}$ 
the Zariski-open subset comprising the \emph{regular} cotangent vectors (centralizer of minimal 
dimension, the rank of $G$).
\begin{theorem}
The following describe the same holomorphic symplectic manifold, denoted $BFM(G)$. \\
(i) The spectrum of the complex equivariant homology $H_*^{G^\vee}(\Omega G^\vee)$, with 
Pontrjagin multiplication. \\
(ii) The holomorphic symplectic reduction of $T^*_{\mathrm{reg}}G_{\mathbb{C}}$ 
by conjugation under $G_{\mathbb{C}}$. \\
(iii) The affine resolution of singularities of the quotient $T^*T_{\mathbb{C}}/W$, 
obtained by adjoining the functions $(e^\alpha-1)/\alpha$. ($\alpha$ ranges over 
the roots of $\mathfrak{g}$, $e^\alpha -1$ is the respective function on $T_{\mathbb{C}}$ 
and the denominator $\alpha$ is the linear function on $\mathfrak{t}^*$.)\\
(iv) $BFM\big(\mathrm{SU}_n\big)$ is the moduli space of $\mathrm{SU}_2$ monopoles of charge $n$, 
and is a Zariski-open subset of the Hilbert scheme of $n$ points in $T^*\mathbb{C}^\times$ \cite{athit}.\\
(v) $BFM(T) = T^*T_{\mathbb{C}}$
\end{theorem}

\begin{remark}\label{centralizer}
The moment map zero-fiber for the conjugation $G_{\mathbb{C}}$-action on $T^*_{\mathrm{reg}}
G_{\mathbb{C}}$ is the (regular) \emph{universal centralizer} $\mathcal{Z}_{\mathrm{reg}} =
\{ (g,\xi)\,|\, g\xi g^{-1} =\xi, \xi \text{ is regular}\}$. 
$\mathcal{Z}_{\mathrm{reg}}$ is smooth, and $BFM(G) = \mathcal{Z}_{\mathrm{reg}}/G_{\mathbb{C}}$, 
with stabilizer of constant dimension and local slices. 
This is the only one of the descriptions that makes the holomorphic symplectic structure evident.
\end{remark}

The space $BFM(G^\vee)$ inherits two projections from $T^*_{\mathrm{reg}}G_{\mathbb{C}}$: 
$\pi_v$, to the space 
$(\mathfrak{g}^\vee)^*_{\mathbb{C}}/G^\vee_{\mathbb{C}}\cong\mathfrak{t}_{\,\mathbb{C}}/W$ of 
co-adjoint orbits, and $\pi_h$, to the conjugacy classes in $G^\vee_{\mathbb{C}}$. 
Both are Poisson-integrable with Lagrangian fibers. The projection $\pi_v$ will have the more 
obvious meaning for gauge theory, capturing the $H^*(BG)$-module structure on fixed-point 
categories. The projection $\pi_h$ is closely related to the restriction to $T$ 
(and to the \emph{string topology} of flag varieties.) 

The symplectic structure on $BFM(G^\vee)$ relates to its nature as (an uncompletion of) 
the second Hochschild cohomology of the $E_2$-algebra $H_*(\Omega G)$.\footnote{Of course, the 
$E_2$ structure is trivial over the complex numbers and the algebra is quasi-isomorphic to 
its underlying dg ring of chains.} In fact, $BFM(G)$ contains the zero-fiber of $\pi_v$,
$Z:= \mathrm{Spec}H_*(\Omega G)$, as a smooth Lagrangian; it comes from the 
part of $\mathcal{Z}_{\mathrm{reg}}$ with nilpotent $\xi$ (cf.~Remark~\ref
{centralizer}). 

Theorem~\ref{e2action} and Lesson~\ref{lessons}.2 sheafify categories with topological 
$G$-action over the formal neighborhood of $Z$. However, it is the entire space $BFM(G^\vee)$ 
which is the correct receptacle for $G$-gauge theory: gauged TQFTs are objects in the 
$2$-category $\sqrt{\mathfrak{Coh}}(BFM(G^\vee))$. Clearly, that requires a 
rethinking of the notion: the definition of `topological category with $G$-action' as in \S2 would 
complete the $BFM$ space at the exceptional Lagrangian $Z$. Loosely speaking, we need to know a
theory together with all its deformations of the group action. 

The Lagrangian $Z$ replaces the zero-section from the torus case, and plays the role 
of the regular representation of $G$: $\Hom(Z,L)$ gives the underlying category of 
the representation $L$. The formal calculation is $\Hom_{C_*\Omega G}(C_*\Omega G; L) = L$, 
if we use Theorem~\ref{e2action} to model representations. Figure~2 below sketches  
$BFM(\mathrm{PSU}_2)$.

\begin{figure}
\includegraphics[height=3in]{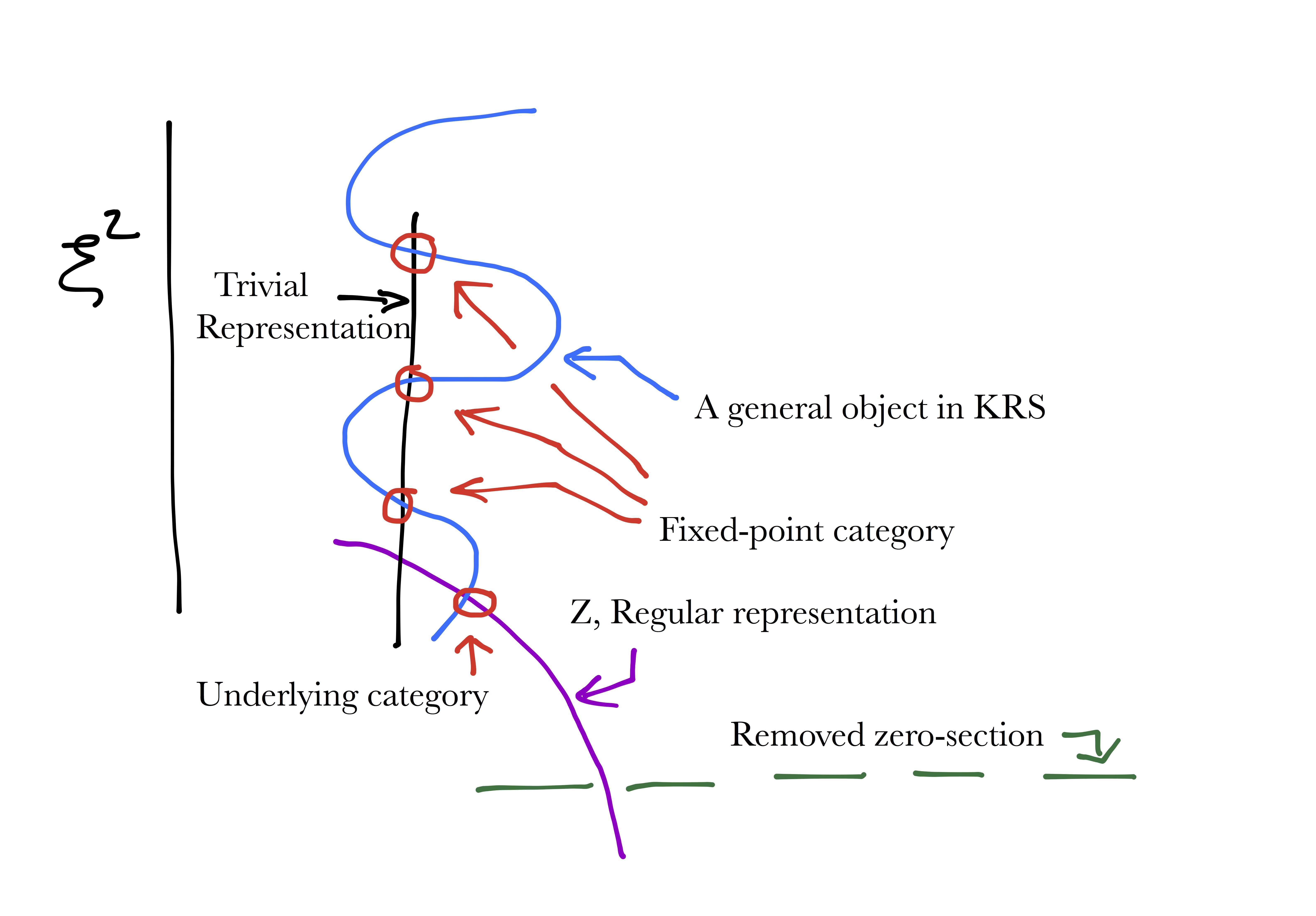}
\caption{$BFM$ space of $\mathrm{PSU}_2$; the fiber of $\pi_h$ at $1$ is 
$Z\cup$trivial representation}  
\end{figure}

\subsection{Induction by String topology} No map relates $BFM(T^\vee) = 
T^*T^\vee_{\mathbb{C}}$ and $BFM(G^\vee)$, because of the blow-up, but a 
holomorphic Lagrangian correspondence is defined from the branched cover
\begin{equation}\label{strintopind}
\xymatrix{
BFM(G^\vee) & BFM(G^\vee)\times_{\mathfrak{t}_{\mathbb{C}}/W} \mathfrak{t}_{\mathbb{C}}\ar[l] \ar[r] 
& T^*T^\vee_{\mathbb{C}}.
}
\end{equation}
The right map is neither proper not open.\footnote{$Z$ maps to $1\in T^\vee$, but most of the 
zero-section in $T^*T^\vee_{\mathbb{C}}$ is missed by the map.} A holomorphic 
Lagrangian correspondences could give a pair of adjoint functors between the respective 
$\sqrt{\mathfrak{Coh}}$ 
$2$-categories, thus a domain wall between $T$- and $G$- gauge theories 
(cf.~\S\ref{boundarycond}). This is indeed the case, and we can identify the functors.

\begin{theorem}
The correspondence~\eqref{strintopind} matches an adjoint pair of restriction-induction functors 
between categorical $T$- and $G$-representations. Induction from a category $\mathfrak{C}$ with 
topological $T$-action is effected by \emph{string topology with coefficients} of the flag variety $G/T$:
\[
\mathrm{Ind}(\mathfrak{C}) = C_*\Omega G \otimes_{C_*\Omega T} \mathfrak{C}.
\]
Restriction is the obvious functor.
\qed
\end{theorem} 

\begin{remark}
(i) An alternative (slightly worse) description of induction is given by the category of (derived) 
global sections $\mathrm{R}\Gamma\big(G/T;\tilde{\mathfrak{C}}\big)$ 
for the associated local system $\tilde{\mathfrak{C}}$ of categories.\\
(ii) Neither description is quite correct. Just as the BFM spaces carry more information 
than the category and the action, so does induction.\\
(iii) For example, inducing from the representation $\mathfrak{Vect}_\tau$, for a point
$\tau\in T^\vee_{\mathbb{C}}$ which is \emph{not} central in $G$, by either method above, will 
appear to give zero. (This is what a homological algebra calculation of the curved string 
topology of $G/T$ for a non-trivial curving $\tau\in H^2(G/T;\mathbb{C}^\times)$ gives.) 
However, geometric induction gives the fiber of $\pi_v^{-1}(\tau)$. 
The puzzle is resolved by noting that none of those fibers meet the regular representation 
$Z$, so the underlying categories are null. We are letting $G$ act on categories without 
objects, and growing wiser.  \\
(iv) The `na\"ively induced' representations can serve to probe the entire $BFM$ space by 
abelianization. It is therefor not \emph{conceptually} more difficult to understand 
non-abelian gauged mirrors than abelian ones. However, the symplectically induced 
representations of the next section are much nicer.
\end{remark}

\subsection{Alternative model for induction} \label{altind}
I close with a new model for the correspondence \eqref{strintopind}, useful in a later 
mirror calculation. Call $\mathfrak{b}_{+,\mathrm{reg}} \subset\mathfrak{b}_+$ the open 
subset of regular elements. Identify $\mathfrak{b}_+ = \big(\mathfrak{g}_{\mathbb{C}}/
\mathfrak{n}_+\big)^*$, $B_+$-equivariantly; the last space matches the fibers of the 
bundle, over $B_+\subset G_{\mathbb{C}}$, of co-normals to the $N_+$-translation orbits. 
Using this to define the left map below and projection on the right gives a holomorphic 
Lagrangian correspondence
\[
\xymatrix{
& \displaystyle{\frac{B_+\times \mathfrak{b}_{+,\mathrm{reg}}}{B_+}} \ar[dl]\ar[dr]& \\
T^*_{\mathrm{reg}}G_{\mathbb{C}}/\!/_{\mathrm{ad}}B_+& & T^*T_{\mathbb{C}} 
}
\] 
having divided by the conjugation action of $B^\vee_+$. 
We can also divide out by $B_+$ in the defining correspondence for $BFM(G)$,
\[
\xymatrix{
  BFM(G) & \mathcal{Z}_{\mathrm{reg}}/B_+ \ar[l]\ar[r] & 
  		T^*_{\mathrm{reg}}G_{\mathbb{C}}/\!/_{\mathrm{ad}}B_+.
}
\]
The composition of these two can be shown to yield \eqref{strintopind} (for the group $G$).

\section{Mirrors of flag varieties}\label{mirflag}
I will now explain the place of flag varieties in the mirror view of gauge theory.  
Lifting to the torus-equivariant picture will recover a construction of 
K.~Rietsch~\cite{r}. 

\subsection{Flag varieties as domain walls}
Let $L\subset G$ be a Levi subgroup, centralizer of a dominant weight $\lambda:\mathfrak{l}\to 
\mathrm{i}\mathbb{R}$.
The flag variety $X=G/L$ is a symplectic manifold with Hamiltonian $G$-action (the co-adjoint 
orbit of $\lambda$), and as such it should have a mirror holomorphic Lagrangian in $BFM(G^\vee)$. 
This will be true, but we forgot some structure relevant to gauge theory. Namely, we can use 
$G/L$ to \emph{symplectically induce} categorical representations from $L$ to $G$. 

A categorical representation $\mathfrak{C}$ of $L$ gives the local system of categories 
$\tilde{\mathfrak{C}} = G\times_L \mathfrak{C}\to X$, and we can construct the Fukaya category 
of $X$ with coefficients in $\tilde{\mathfrak{C}}$. (Objects would be 
horizontal sections of objects over Lagrangians, and Floer complexes can be formed in the usual 
way from the $\Hom$-spaces over intersections.) In fact, the weight $\lambda$ (or rather, 
its exponential $e^\lambda$ in the center of $L^\vee_{\mathbb{C}}$) defines a 
topological representation $\mathfrak{Vect}_\lambda$ of $L$, and we can think of the ordinary 
Fukaya category $\mathfrak{F}(X,\lambda)$ as the symplectic induction from the latter. The precise meaning 
is that deforming $\lambda$ in $\mathfrak{Vect}_\lambda$ achieves the same effect as the matching 
deformation of the symplectic form. An imaginary variation of $\lambda$ (movement in the unitary 
group $L^\vee$) has the effect of adding a unitary $B$-field twist to the Fukaya category.

\begin{remark}
Left adjoint to the symplectic induction functor $\mathrm{SInd}_L^G$ is a \emph{symplectic restriction} 
from $G$ to $L$. This is not the ordinary (forgetful) restriction, which instead is adjoint 
to string topology induction (\S5). For example, when $L=T$, the spectral decomposition under $T$ of the 
symplectic restriction of $\mathfrak{C}$ would extract the multiplicities of the $\mathfrak{F}(X,\tau)$ 
in $\mathfrak{C}$, rather than those of the $\mathfrak{Vect}_\tau$. 
\end{remark}

This pair of functors is a new \emph{domain wall} between pure $3$-dimensional $G$- and $L$-gauge 
theories. On the mirror side, we can hope to represent a domain wall  by a holomorphic Lagrangian 
correspondence between $BFM(L^\vee)$ and $BFM(G^\vee)$. 
We will be fortunate to identify this correspondence as an open embedding. 

To recover the mirror of $X$ in its various incarnations (as a symplectic manifold, or 
a $G$-equivariant symplectic one) we must apply boundary conditions to the two gauge 
theories, aiming for the `sandwich picture' of a $2D$ TQFT, as in \S\ref{boundarycond}.
For example, to find the underlying symplectic manifold $(X,\lambda)$, we must apply the 
representation $\mathfrak{Vect}_\lambda$ of $L$ and the regular representation $Z$ of $G$. 
I shall carry out this (and a more general) exercise in the final section. 
 
The study of symplectically induced representations can be motivated by the following 
conjecture, the evident non-abelian counterpart of Conjecture~\ref{torusgit} (with the 
difference that it seems much less approachable).

\begin{conjecture}
For a Hamiltonian $G$-action on the compact symplectic manifold $X$ and a regular value $\mu$ 
of the moment map, the Fukaya category $\mathfrak{F}(X/\!/G)$, reduced at the orbit of $\mu$ 
(and with unitary $B$-field $\mathrm{i}\nu$) is the multiplicity in $\mathfrak{X}$ of the 
representation symplectically induced from $\mathfrak{Vect}_{\mu+i\nu}$.
\end{conjecture}

\subsection{The Toda isomorphism} The following isomorphism of holomorphic symplectic manifolds 
is mirror to symplectic induction. It fits within a broad range of related results (`Whittaker 
constructions') due to Kostant. Its relation to Fukaya categories of flag varieties is mysterious,  
and now only understood with reference to the appearance of the Toda integrable system in 
the Gromov-Witten theory of flag varieties \cite{gk, kos}. From that point of view, the 
isomorphism enhances the Toda system by supplying the conjugate family of commuting Hamiltonians, 
pulled back from conjugacy classes in the group, rather than orbits the Lie algebra. 

The mirror picture of $G$-gauge theory involves the Langlands dual group $G^\vee$ of $G$, but 
the notation is cleaner with $G$. With notation as in \S5, call $\chi: \mathfrak{n}\to 
\mathbb{C}^\times$ the regular character (unique up to $T_{\mathbb{C}}$-conjugation) and 
consider the \emph{Toda space}, the holomorphic symplectic quotient of $T^*G_{\mathbb{C}}$
\[
T(G): = (N,\chi)\backslash\!\backslash T^*G_{\mathbb{C}} /\!/(N,\chi)
\]  
under the left$\times$right action of $N$, reduced at the point $(\chi,\chi)\in 
\mathfrak{n}^*\oplus\mathfrak{n}^*$. 
\begin{theorem}\label{todaisom}
We have a holomorphic symplectic isomorphism
\[
T(G) = (N,\chi)\backslash\!\backslash T^*G_{\mathbb{C}} /\!/(N,\chi) 
	\cong T^*_{\mathrm{reg}}G_{\mathbb{C}}/\!/_{\mathrm{Ad}}G_{\mathbb{C}}= BFM(G)
\]
induced from the presentation of the two manifolds as holomorphic symplectic reductions 
of the same manifold $T^*_{\mathrm{reg}}G_{\mathbb{C}}$. 
\end{theorem}

\begin{proof}
The $N\times N$ moment fiber in $T^*G_{\mathbb{C}}\cong G_{\mathbb{C}}\times 
\mathfrak{g}_{\mathbb{C}}^*$ (by left trivialization) is 
\[
\mathcal{T}:=\{(g,\xi)\in G_{\mathbb{C}}\times \mathfrak{g}_{\mathbb{C}}^*\, |\, \pi(\xi) = 
	\pi(g\xi g^{-1})=\chi\},
\]
where $\pi:\mathfrak{g}_{\mathbb{C}}^*\to\mathfrak{n}^*$ is the projection. As $\pi^{-1}(\chi)$ 
consists of regular elements, we may use $T^*_{\mathrm{reg}}G_{\mathbb{C}}$ instead. 
Now, $N$ acts freely on $\pi^{-1}(\chi)$, with Kostant's global slice, so the $N\times N$ 
action on $\mathcal{T}$ is free also and $T(G)= N\backslash\mathcal{T}/N$ is a manifold. 

The moment map fibers $\mathcal {T}$ and $\mathcal{Z}_{\mathrm{reg}}$ (for the $\mathrm{Ad}$-action of 
$G_{\mathbb{C}}$) provide holomorphic Lagrangian correspondences
\begin{equation}\label{todabfm}
\xymatrix{
& \mathcal{T} \ar[dl]\ar[dr]& & \ar[dl]\mathcal{Z}_{\mathrm{reg}} \ar[dr]& \\
T(G) & & T^*_{\mathrm{reg}}G_{\mathbb{C}} & & BFM(G)
}
\end{equation}
whose composition $\mathcal{T} \times_{T^*_{\mathrm{reg}}G_{\mathbb{C}}} \mathcal{Z}_{\mathrm{reg}}$, 
I claim, induces an isomorphism. Actually, the clean  correspondence must mind the fact 
that the two actions on $T^*G$, of $N\times N$ and $G$, respectively, have  in common the conjugation 
action of $N$ (sitting diagonally in $N\times N$): so we must really factor through $T^*_{\mathrm{reg}}
G_{\mathbb{C}}/\!/_{\mathrm{Ad}}(N)$, within which the co-isotropics $\mathcal{T}/_{\mathrm{Ad}}N$ 
and $\mathcal{Z}_{\mathrm{reg}}/_\mathrm{Ad}N$ turn out to intersect transversally.

We check that the composition in \eqref{todabfm} induces a bijection on points: preservation 
of the Poisson structure then supplies the Jacobian criterion. Choose $(g,\xi)\in \mathcal{T}$; 
then, $\xi, g\xi g^{-1}\in \pi^{-1}(\chi)$ are in the same $G_{\mathbb{C}}$-orbit 
in $\mathfrak{g}^*_{\mathbb{C}}$. Kostant's slice theorem ensures that the two elements are then 
$\mathrm{Ad}$-related by a unique $\nu\in N$, $\nu g\xi (\nu g)^{-1} =\xi$. There is then, up to 
right action of $N$, a unique $(g',\xi')\in \mathcal{Z}_{\mathrm{reg}}$ in the $N\times N$-orbit 
of $(g,\xi)$. We thus get an injection $T(G)\hookrightarrow BFM(G)$. To see surjectivity, conjugate 
a chosen $(h,\eta)\in \mathcal{Z}_{\mathrm{reg}}$ to bring $\eta$ into $\pi^{-1}(\chi)$. The 
result is in $\mathcal{T}$ (and is again unique up to $N$-conjugation).  
\end{proof}

\begin{remark}
The space $T(G)$ has a hyperk\"ahler structure; it comes from a third description, as a moduli 
space of solutions to Nahm's equations. This is closely related to a conjectural derivation 
of my mirror conjecture \eqref{key2} below from Langlands (electric-magnetic) duality in 
$4$-dimensional $N=4$ Yang-Mills theory. (I am indebted to E.~Witten for this explanation.)
\end{remark}

\subsection{The mirror of symplectic induction}
Inclusion of the open cell $N\times w_0\cdot T_{\mathbb{C}}\times N\subset G_{\mathbb{C}}$ 
leads to a holomorphic symplectic embedding $T^*T_{\mathbb{C}} \subset T(G)$. Sending a co-tangent vector 
to its co-adjoint orbit projects $T(G)$ to $\mathfrak{g}^*_{\mathbb{C}}/\!/G^{ad}_{\mathbb{C}}$, and 
the functions on the latter space lift to the commuting Hamiltonians of the Toda integrable system; 
so the theorem completes the picture by providing a complementary set of Hamiltonians lifted 
from the conjugacy classes of $G$. 

More generally, if $L\subset G$ is a Levi subgroup, with representative $w_L\in L$ of its longest 
Weyl element, and with unipotent group $N_L = N\cap L_{\mathbb{C}}$, then $\chi$ restrict to a regular 
character of $N_L$ and the inclusion 
\[
N\times_{N_L} w_0w_L^{-1}\cdot L_{\mathbb{C}} \times_{N_L} N \subset G_{\mathbb{C}}
\]
determines an open embedding $T(L)\subset T(G)$. The following is, among others, a character formula 
for induced representations. It relies on too many wobbly definitions to be called a theorem, but 
assuming it is meaningful, its truth can be established form existing knowledge. 

\begin{conjecture}\label{key2}
Via the Toda isomorphism, the embedding $T(L^\vee)\subset T(G^\vee)$ is mirror to symplectic 
induction from $L$ to $G$, representing the flag variety $G/L$ as a domain wall between $L$- and 
$G$-gauge theories. 
\end{conjecture}
\begin{example}
With the torus $L=T$, a one-dimensional representation of $T$ is described by a point 
$q\in T^\vee$, represented in $\sqrt{\mathfrak{Coh}}(T^*T^\vee)$ by the cotangent space at $q$. 
Its image under the Toda isomorphism, a Lagrangian leaf $\Lambda(q)\subset BFM(G)$, 
is the symplectically induced representation, or the $G$-equivariant Fukaya category 
of the flag variety $G/T$ with quantum parameter $q$. The analogue of the character 
is the structure sheaf $\mathcal{O}_{\Lambda(q)}$, whose algebra of global sections 
is the $G$-equivariant quantum cohomology of $G/T$ \cite{gk}.
\end{example}  
\begin{remark}
It is difficult to prove the conjecture without a precise definitions (of equivariant 
Fukaya categories with coefficients and of the $KRS$ $2$-category). Nevertheless, accepting 
that $BFM(G^\vee)$ as the correct mirror of $G$-gauge theory, the conjecture follows from 
known results about the equivariant quantum cohomology of flag varieties \cite{gk, cf, mih}. 
The latter describe $qH^*_G(G/L)$ as a module over 
$H^*(BG)= \mathbb{C}[\mathfrak{g}]^G$, the algebra of Toda Hamiltonians, induced from 
the projection $\pi_v$. The symplectic condition turns out to pin the map  uniquely.
\end{remark}

\subsection{Foliation by induced representations}
Recall (Example~\ref{cuspidalreps}) the one-dimensional representations of a Levi subgroup 
$L\subset G$, corresponding to the points in the center of $L^\vee_{\mathbb{C}}$. Let us 
call them \emph{cuspidal}: they are not symplectically induced from a smaller Levi subgroup.  
(Such a symplectic induction produces representation of rank equal to the Euler characteristic 
of the flag variety.) The following proposition suggests that these induced representations 
are better suited to spectral theory that the na\"ively induced ones 
of \S\ref{bfmsect}.

\begin{theorem}
The space $BFM(G^\vee)$ is smoothly foliated by symplectic inductions of cuspidal representations: 
each leaf comes from  a unique cuspidal representation of a unique Levi subgroup $L$, 
with $T\subset L \subset G$.
\end{theorem}
\begin{proof}
The leaves are the fibers of $N\backslash\mathcal{T}^\vee/N\to N\backslash G^\vee_{\mathbb{C}}/N$, 
and induction on the semi-simple rank reduces us to checking that the part of $ \mathcal{T}^\vee$ 
which does \emph{not} come from any $T(L^\vee)$, for a proper $L\subset G$, lives over the 
center of $G^\vee_{\mathbb{C}}$.  

Omit $^\vee$ from the notation and choose $(g,\xi)\in \mathcal{T}$. From 
$G_{\mathbb{C}} =\coprod_{w} N\cdot wT_{\mathbb{C}}\cdot N$, we may take $g \in 
wT_{\mathbb{C}}$ for some $w\in W$. Split 
$\mathfrak{g}^*_{\mathbb{C}} = \mathfrak{n}^*\oplus\mathfrak{t}^*_{\mathbb{C}} \oplus \mathfrak{n}^*_+$; 
then, 
\begin{align}
\xi &= \chi + \eta +\nu, &\text{for some } \eta\in \mathfrak{t}_{\mathbb{C}},\: &\nu\in \mathfrak{n}_+^* 
	\nonumber\\
g\xi g^{-1} &= \chi + w(\eta) + \nu', &\text{for some }\nu'\in \mathfrak{n}_+^* &\nonumber
\end{align} 
whence we see that $w$ sends each simple negative root either to a simple negative root, or 
to a positive root. If $w=1$, then $g\in T_{\mathbb{C}}$ centralizes $\chi \pmod{\mathfrak{b}^*_+}$ 
and thus lies in the center of $G_{\mathbb{C}}$. Otherwise, I claim that $w=w_0w_L^{-1}$, for the Levi 
$L$ whose negative simple roots stay negative. Equivalently, the unique simple root system 
of $\mathfrak{g}$ comprising the simple negative roots of $L$ and otherwise only positive roots,  
is the $w_L$-transform of the positive root system. This can be seen by choosing a point 
$\zeta + \varepsilon$, with $\zeta$ generic on the $L$-fixed face of the dominant Weyl chamber, 
and $\varepsilon$ a dominant regular displacement: $w_L(\zeta+\varepsilon)$ must be in the dominant 
chamber of the new root system.      
\end{proof}

\begin{example}[$G=\mathrm{SU}_2$]
The dual complex group is $G^\vee_{\mathbb{C}}=\mathrm{PSL}_2(\mathbb{C})$, whose $BFM$ space 
is the blow-up of $\mathbb{C}\times \mathbb{C}^\times/\{\pm 1\}$ at $(0,1)$, 
with the proper transform of the zero-section $\{0\}\times \mathbb{C}^\times/\{\pm 1\}$ 
removed. This is the Atiyah-Hitchin manifold studied in \cite{athit}. 
The $\mathbb{Z}/2$-action identifies $(\xi, z)$ with $(-\xi,z^{-1})$. Projection 
to the line of co-adjoint orbits is given by the Toda Hamiltonian $\xi^2$. 

The Toda inclusion of $T^*T^\vee_{\mathbb{C}}\cong \mathbb{C}\times\mathbb{C}^\times$ sends 
a point $(u,q)$ to 
\[
\xi^2 = u^2-q,\qquad  
\frac{z+z^{-1}}{4} = \frac{u^2}{q} -\frac{1}{2}
\]
(A match of signs is  required between $z$ and $\xi$.) The induced 
leaves of constant $q$ are given by  
\[
\xi = q\frac{\sqrt{z} -\sqrt{z}^{-1}}{2},
\] 
after lifting to the coordinates $\xi, \sqrt{z}$ for the double-cover maximal torus in 
$\mathrm{SL}_2$. We recognize here the (graph of the differentiated potential in the) 
$S^1$-equivariant mirror of the flag variety $\mathbb{P}^1$. 

The one remaining leaf in $BFM(\mathrm{PSU}_2)$ is the trivial representation of $\mathrm{SU}_2$;
it is the proper transform of $T^*_1\mathbb{C}^\times/\{\pm1\}$, the image in $\mathbb{C}\times 
\mathbb{C}^\times/\{\pm 1\}$ of) the cotangent fiber at $1$. If we switch instead to $\mathrm{PSU}(2)$,  
the new $BFM$ space (on the Langlands dual side) is a double cover of the former, and there is a 
new cuspidal leaf over the central point $(-\mathrm{I}_2)\in \mathrm{SU}_2$, corresponding to the 
sign representation of $\pi_1\mathrm{PSU}_2$. 
\end{example}

\subsection{Torus-equivariant flag varieties} Restricting the $G$-action to $T$, 
the flag manifold $G/L$ is a transformation from $L$-gauge theory to $T$-gauge theory, 
given by composition of the symplectic induction and string topology domain walls: 
\begin{equation}\label{compos}
\xymatrix{
T(L^\vee) \ar@{^{(}->}[r]^{\mathrm{\:SInd\:\:}} & T(G^\vee) \ar[r]_<<<<{\mathrm{Toda}}^<<<<{\sim}& BFM(G) 
\ar[r]^<<<<{\mathrm{ST}} 	& BFM(T^\vee) =T^*T^\vee_{\mathbb{C}}
} 
\end{equation}
The equivariant mirror is a family of $2D$ TQFTs, which can be defined, for instance, by a 
family of complex manifolds with potentials parametrized by the Lie algebra $\mathfrak{t}_{\,\mathbb{C}}$. 
This family reflects the $H^*(BT)$-module structure on equivariant quantum cohomology. 
When $\mathfrak{F}(G/L)$ has been represented by an object $\Lambda\in \sqrt{\mathfrak{Coh}}
(T^*T^\vee_{\mathbb{C}})$, 
the family comes from the projection of $T^*T^\vee_{\mathbb{C}}$ to the cotangent fiber, 
and the TQFTs are the fibers of $\Lambda$ over $\mathfrak{t}_{\,\mathbb{C}}$, the 
$\Hom$ categories with the constant sections of $T^*T^\vee_{\mathbb{C}}$.  

To recover this family of mirrors from the double domain wall~\eqref{compos}, we must use it 
to pair two Lagrangians, in $T(L^\vee)$ and in $T^*T^\vee_{\mathbb{C}}$. The Lagrangians are

\begin{itemize}
\item the Lagrangian leaf $\Lambda(q)\subset BFM(L^\vee)$ over a point $q$ in the 
center of $L^\vee_{\mathbb{C}}$, describing a cuspidal representation of $L$  
($q$ is also the quantum parameter for $G/L$); 
\item the constant Lagrangian section $S_\xi$ of $T^*T^\vee_{\mathbb{C}}$, with fixed value 
$\xi\in \mathfrak{t}_{\,\mathbb{C}}$. 
\end{itemize}

Note that $S_\xi$ is the differential of a multi-valued character $\xi\circ\log: T^\vee_{\mathbb{C}}\to 
\mathbb{C}$.

\begin{remark}
The relevant TQFT picture is a sandwich with triple-decker filling: 
the base slice is the representation $\mathfrak{Vect}_q$ of $L$ corresponding to $\Lambda(q)$, 
a boundary condition for $L$-gauge theory. The filling of the sandwich is a triple layer of 
$L,G,T$ gauge theories, separated by the $\mathrm{SInd}$ and string topology domain walls in 
\eqref{compos}. The sandwich  is topped with the slice $S_\xi$, a boundary condition for 
$T$-gauge theory. Its underlying representation category is null, if $\xi\neq 0$; $S_\xi$ 
is a deformation of the regular representation of $T$ by the multi-valued potential $\xi\circ\log$. 
\end{remark}

\subsection{Rietsch mirrors}
Building on ideas of Peterson and earlier calculations of Givental-Kim, Ciocan-Fontanine, 
Kostant and Mihalcea \cite{gk, cf, kos, mih}, Rietsch \cite{r} proposed torus-equivariant 
complex mirrors for all flag varieties $G/L$. 

Let us recover these from my story by computing the answer outlined above. 
Recall (\S\ref{altind}) the Lagrangian correspondence 
\[
T^*T_{\mathbb{C}} \leftarrow B_+\times \mathfrak{b}_{+,\mathrm{reg}} \to T^*_\mathrm{reg}G_{\mathbb{C}},
\]
appearing in the alternate model for the string topology induction. Compose this with the Toda 
construction to define the following holomorphic Lagrangian correspondence 
between $T(G)$ and $BFM(T) = T^*T_{\mathbb{C}}$: 
\begin{equation}\label{rietsch}
\xymatrix{
& \mathcal{T}\ar[dl]_P\ar[dr]  & & B_+\times \mathfrak{b}_{+,\mathrm{reg}}\ar[dl] \ar[dr]^p& \\
T(G) & & T^*_{\mathrm{reg}}G_{\mathbb{C}} & &  T^*T_{\mathbb{C}}
}
\end{equation}

\begin{proposition}\label{altcor}
Correspondence \eqref{rietsch} is the composition $\mathrm{ST}\circ \mathrm{Toda}$ 
of \eqref{compos}.  
\end{proposition}

\begin{proof}[Sketch of proof.] In the jagged triangle of correspondences below, 
the left edge is the Toda isomorphism, the right edge the correspondence \eqref{rietsch} 
and the bottom edge the string topology domain wall. The long, counterclockwise way from 
top to right involves division by the complementary subgroups $N$ and $B_+$ of $G_{\mathbb{C}}$; 
so it seems reasonable that the composition should agree with the undivided 
correspondence~\eqref{rietsch} on the right edge: 
\begin{equation*}
\xymatrixcolsep{.2cm}
\xymatrix{
 & \mathcal{T}/_\mathrm{Ad}(N)\ar[r]\ar[d] & \framebox{T(G)} & \ar[l] \mathcal{T}\ar[d] & 
 			\mathcal{T}\cap \big(B_+\times \mathfrak{b}_{+,\mathrm{reg}}\big) \ar@{.>}[d]\ar@{.>}[l]\\
\mathcal{Z}_{\mathrm{reg}}/_\mathrm{Ad}(N)\ar[dd]\ar[r] & T^*_{\mathrm{reg}}G_{\mathbb{C}}/\!/_{\mathrm{Ad}} N & 
		\mathcal{Z}_{\mathrm{reg}}(B_+)/B_+ \ar@{.>}[dl]\ar@{.>}[dr]
	& T^*_{\mathrm{reg}}G_{\mathbb{C}} & B_+\times \mathfrak{b}_{+,\mathrm{reg}}\ar[dd]\ar[l] \\
	& \displaystyle{\frac{\mathcal{Z}_{\mathrm{reg}}}{B_+}} \ar[dl]\ar[dr] & 
			& \displaystyle{\frac{B_+\times\mathfrak{b}_{+,\mathrm{reg}}}{B_+}} \ar[dl]\ar[dr]& \\
\framebox{BFM(G)} & & T^*_{\mathrm{reg}}G_{\mathbb{C}}/\!/_{\mathrm{Ad}} B_+ & & \framebox{BFM(T)}
}
\end{equation*}

The argument exploits the regularity of the Lie algebra elements. 
The intersection in the upper right corner comprises the pairs $(b,\beta)\in B_+\times 
\mathfrak{b}_+$ with $b$ centralizing $\beta\in E+ \mathfrak{t}_{\,\mathbb{C}}$. 
($E=\chi$ under $\mathfrak{n}_+\cong \mathfrak{n}^*$.)  That is a slice for the conjugation 
$B_+$-action on the regular centralizer $\mathcal{Z}_{\mathrm{reg}}(B_+)$ in $B_+$, which makes 
clear the isomorphism with the fiber product in the center the triangle; and the map is 
compatible with the Toda isomorphism on the left edge. 
\end{proof}

We now calculate the pairing $S_\xi\subset T(L^\vee)$  and 
$\Lambda(q)\subset T^*T^\vee_{\mathbb{C} }$ by the correspondence~\eqref{rietsch} for 
the dual group $G^\vee$. We do so by computing in $T^*_{\mathrm{reg}}G^\vee_{\mathbb{C}}$
\[
\Hom\big(p^{-1} S_\xi,P^{-1} \Lambda(q)\big).
\]  
The two Lagrangians meet over the intersection 
\[
B^\vee_+\cap \big(N^\vee\cdot w_0w_L^{-1}L^\vee_{\mathbb{C}}\cdot N^\vee\big) \subset G^\vee_{\mathbb{C}}. 
\]
Lift $\xi\circ\log$ to $B^\vee_+$ by $p$; over $B^\vee_+$, $p^{-1}S_\xi$ is the conormal bundle to 
$B^\vee_+\subset G^\vee_{\mathbb{C}}$ shifted by the graph of $d (\xi\circ\log)$. (The shifted 
bundle is well-defined, independently of any local extension of the function $\xi\circ\log$.)   

The Lagrangian $P^{-1}\Lambda(q)$ lives over the open set $N^\vee\cdot w_0w_L^{-1}L^\vee_{\mathbb{C}}
\cdot N^\vee$ in $G^\vee_{\mathbb{C}}$, where it is the shifted co-normal bundle to the submanifold 
\[
M:= N^\vee \cdot w_0w_L^{-1}q\cdot N^\vee \cong \frac{N^\vee\times N^\vee}
{\mathrm{diag}(N^\vee\cap L^\vee_{\mathbb{C}})},
\] 
shifted into $\mathcal{T}$ by the graph of the differential of the following function $f$: 
\[
f: n_1\cdot w_0w_L^{-1} l\cdot n_2 \mapsto \chi (\log n_1 + \log n_2).
\]  

Now, $B_+^\vee$ and $M$ meet transversally in $G^\vee_{\mathbb{C}}$, in a manifold 
isomorphic to a Zariski-open in the flag variety $G^\vee/L^\vee$; this is the $\mathcal{R}_{w_0,w_L}$ 
of \cite{r}-. Transversality permits us to dispense with the conormal bundles, and identify 
$\Hom(S_\xi, \Lambda(q))$ with the pairing, in the cotangent bundles, between graphs of the 
restricted functions to $B^\vee_+\cap M$
\[
\Hom_{T^*(B^\vee_+\cap M)}\left(\Gamma(d(\xi\circ\log)), \Gamma(df)\right); 
\]
this is the matrix factorization category $MF\big(B^\vee_+\cap M; f-\xi\circ\log\big)$. This 
is the Rietsch  mirror of $G/L$.

The last mirror comes with a volume form, which defines the trace 
on $HH_*$. In the Lagrangian correspondence, we need instead a half-volume form on each leaf. 
The two leaves $S_\xi$ and $\Lambda(q)$ do in fact carry natural half-volumes, translation-invariant 
for the groups ($B$ and $N\times N$) and along the cotangent fibers. Rietsch's volume form on the 
mirror $\mathcal{R}_{w_0,w_L}$ comes from the product of these half-volumes.


\small{Department of Mathematics, UC Berkeley, CA 94720, USA}\\
\texttt{teleman@berkeley.edu}

\end{document}